\documentclass[letterpaper, 10 pt, conference]{ieeeconf}
\IEEEoverridecommandlockouts                                                        
\usepackage{amssymb}
\usepackage{amsmath}
\usepackage{amsfonts}
\usepackage{graphicx}

\newtheorem{theorem}{Theorem}

\newtheorem{lemma}[theorem]{Lemma}

\newtheorem{proposition}[theorem]{Proposition}
\newtheorem{remark}[theorem]{Remark}

\newtheorem{assumption}[theorem]{Assumption}

\title{Dual Rate Control for Security in Cyber-physical Systems }
\author{Mohammad Naghnaeian, Nabil Hirzallah and Petros G. Voulgaris \thanks{%
M. Naghnaeian is a PhD candidate with the Mechanical Science and Engineering
Department, University of Illinois, Urbana, IL, USA \texttt{\small %
naghnae2@illinois.edu}} \thanks{%
N. Hirzallah is a PhD candidate with the Electrical and Computer Engineering
Department, University of Illinois, Urbana, IL, USA \texttt{\small %
hirzall2@illinois.edu}}\thanks{%
P. G. Voulgaris is with the Aerospace Engineering Department and the
Coordinated Science Laboratory, University of Illinois, Urbana, IL, USA
\texttt{\small voulgari@illinois.edu}} \thanks{%
This work was supported in part by the National Science Foundation under NSF
Award NSF ECCS 10-27437 and AFOSR under award AF FA 9550-12-1-0193}}

\begin{document}

\maketitle
\thispagestyle{empty}
\pagestyle{empty}

\begin{abstract}

We consider malicious attacks on actuators and sensors of a feedback system which can be modeled as additive, possibly unbounded, disturbances at the digital (cyber) part of the feedback loop.  We precisely characterize the role of the unstable poles and zeros of the system in the ability to detect stealthy attacks in the context of the sampled data implementation of the controller  in feedback with the continuous (physical) plant.  We show that, if there is a single sensor that is guaranteed to be secure and the plant is observable from that sensor, then there exist a class of multirate sampled data controllers that ensure that all attacks remain detectable.  These dual rate controllers are sampling the output faster than the zero order hold rate that operates on the control input and as such, they can even provide better nominal performance than single rate, at the price of higher sampling of the continuous output.

\end{abstract}
\section{Introduction}

Security of cyber-physical systems has caught a lot of attention lately. Recent papers along with successful attacks on critical infrastructure together revealed many vulnerabilities in the practiced methods of control. For instance, \cite{Ning} showed that if a hacker can access the cyber-space of the power grid, then it is easy for him to change the power state estimates without being detected by the traditional bad data detection methods provided that he knows the grid configuration. This led to many research papers investigating the security of the state estimates and suggesting protective measures in addition to investigating attacks on the actuators and/or the plant itself. For example, \cite{sandberg1}, \cite{sandberg2} introduce security indices which quantify the minimum effort needed to change the state estimates without triggering bad-data detectors with perfect and imperfect knowledge of the system as constraints. In \cite{tabuada} the authors considered attacks on control system measurements that are not necessarily bounded or a follow a certain distribution and without prior knowledge of the system. They show that it is impossible to reconstruct the states of the system if more than half of the sensors are attacked, generalizing some earlier results in \cite{hadj}. However, an NP-hard problem has to be solved to detect the attacks. In \cite{sinopoli,newsino}, the authors inject a signal (unknown to the attacker) into the system to detect replay attacks at the expense of increasing the cost of the LQG controller. However, if the plant has an unstable zero then it can be shown that an undetectable attack can still be designed. In \cite{bullo}, the authors suggest the use of dynamic filters that continuously monitor the states of the system at every instance of time. However, the filters have a serious limitation in that they cannot detect zero dynamics attacks. In \cite{sandberg3}, the authors investigate the class of zero dynamics attacks and suggest adding extra sensors or even perturbing the plant by adding extra connections to remove the unstable zeros. However this may not always be feasible in practice.

In this paper we focus on attacks on actuators and sensors, represented as additive and unbounded disturbances on the digital (i.e., ``cyber") part of the controlled system.  We examine from an input-output perspective the exact conditions under which such attacks can be stealthy, which brings up the pivotal role of unstable zeros and poles of the open loop, continuous time, physical plant.  A key point that the paper brings is the sampled-data (SD) nature of a controlled cyberphysical system which consists of the continuous physical dynamics and the digital controller.  The importance of the SD nature lies in the fact that typically, to ensure good intersample behavior, the rate of the sample and hold mechanism has to be high enough. It is known however that high sampling rate can lead to unstable zeros in the discrete plant dynamics.   In particular, even if a continuous LTI plant $P_c$  has no unstable zeros, its discrete representation $P_d$ obtained by the sampled and hold operations will introduce unstable zeros if the relative degree of $P_c$ is greater than three (e.g., \cite{chenfranbook}) and the sampling period $T\rightarrow0$ (see Figure \ref{SDfig}.)  Therefore, a SD implementation of the controller may create additional vulnerability to stealthy attacks and so, it is important to have ways that secure the safety of the system while achieving the required performance.  As one such way, we propose a dual rate sampling approach, a special case of multirate sampling (MR), whereby the output is sampled at a multiple of the hold rate.

Multirate sampling has been studied extensively in the context of sampled-data control
in the past and many relevant analysis and synthesis results were
obtained in the mid 80s to mid 90s era (e.g., \cite{hagiwara,voudah94,voubam93,chenqi,qichen,cola,meyer,khargo} to mention only a few.) An
interesting property of multi-rate sampling is its ability to remove certain unstable
zeros of the discrete-time system when viewed in the lifted LTI domain, which in turn allows for fulfilling certain
potential design requirements such as gain margin levels, or, strong
stabilization, that are not possible to satisfy with single rate.  It is precisely this property that we utilize and study in detail in the context of stealthy attack detection.  We show that dual rate control is sufficient to remove all the vulnerabilities to stealthy actuator attacks.  Of course, if all sensors are attached as well then there is no way to detect attacks.  On the other hand, we show that if a single measurement output remains secure, and if the modes of the system are observable from this output, then dual rate systems always provide the ability to detect combined sensor-actuator attacks.

Some standard notation we use is as follows: $\mathbb{Z}_{+}$, $\mathbb{R}^{n}$, $\mathbb{C}^{n}$ and $\mathbb{R}^{n\times m}$ denote the sets of non-negative integers, $n$-dimensional
real vectors, $n$-dimensional
complex vectors and $n\times m$ dimensional real matrices, respectively. For any $\mathbb{R}^{n}$ or $\mathbb{C}^{n}$ vector $x$ we denote $x'$ its transpose and $|x|:=\max_i\sqrt{x_i^2}$ where  $x'=\left[ x_{1},x_{2},...,x_{n}\right]$; for a sequence of real $n$-dimensional
vectors, $x=\{x(k)\}_{k\in\mathbb{Z}_{+}} $ we denote $||x||:=\sup_k|x(k)|$; for a sequence of real $n\times m$ dimensional real matrices $G=\{G_k\}_{k\in\mathbb{Z}_{+}} $ we denote its $\lambda$-transform $G(\lambda):=\sum_{k=0}^{\infty}G_k\lambda^k.$ For a $\lambda$-transform $x(\lambda)$ of a sequence $x$ of $n$-dimensional vectors $||x(\lambda)||=||x||$.

\section{System Model}
\begin{figure}[t]
\centering
\includegraphics
[width=3in]
{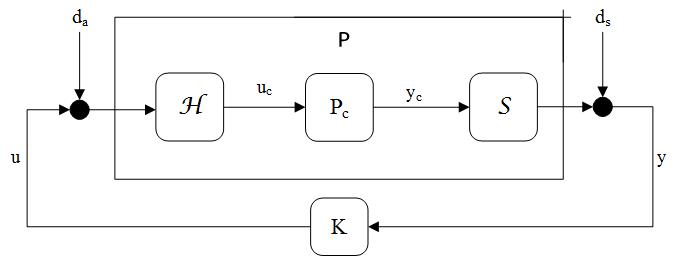}
\caption{The standard SD system}
\label{SDfig}
\end{figure}
We consider the physical, continuous-time, LTI plant $P_c=[A_c, B_c, C_c,
D_c]$ of Figure \ref{SDfig} that is controlled by a digital controller $K$ using
the standard zero order hold and sampling devices $\mathcal{H}$ and $%
\mathcal{S}$ respectively . In particular, in the absence of any disturbances $d_a$ and $d_s$, the digital controller input $u=\{u(k)\}$
converts to the continuous time input $u_c(t)=(\mathcal{H} u)(t)=u(k)$ for $%
kT\le t<(k+1)T$ where $T$ is the hold period, and the digital output $%
y=\{y(k)\}$ sequence is obtained by sampling the continuous time output $y_c$
with the same period $T$, i.e., $y(k)= (\mathcal{S}y_c)(k)=y_c(kT). $ The
corresponding discrete time LTI plant $P$ is defined by the relation $y=Pu$, i.e., $P=\mathcal{S} P_c \mathcal{H}$, and has a description $%
P=[A_d,B_d,C_d,D_d]$ where the state space matrices are obtained from the
corresponding continuous time as
\begin{equation}
\!\begin{aligned}
\label{discabc}
A_d& :=e^{A_{c}T}\in \mathbb{R}^{n\times n},& B_d& :=\int_{0}^{T}e^{A_{c}\tau }B_{c}{\rm d}\tau \in \mathbb{R}^{n\times n_{u}}, \\
C_d& :=C_{c}\in \mathbb{R}^{n_{y}\times n}, & D_d& :=D_{c}\in \mathbb{R}^{n_{y}\times n_{u}}.
\end{aligned}
\end{equation}We assume that the employed realization of the continuous plant $P_c$ is minimal, which implies that the same holds true for the discrete plant $P$ in the absence of pathological sampling (e.g., \cite{chenfranbook},) i.e., for almost all periods $T$.

Also in this figure, we consider the possibility of attacks in terms of
additive disturbances $d_a$ and $d_s$ respectively at the digital input $u$
and at the output $y$ of $P$. These attacks on the digital part of the
system can be on actuators only ($d_s=0$), sensors only ($d_a=0$), or on
both, coordinated or not. As they act on the cyber part of the system we
allow them to be unbounded sequences.

We assume that there is an attack detection mechanism in place that monitors
$u$ and $y$ and can detect an attack only if the effect of $d_{a}$ and/or $%
d_{s}$ on these signals is beyond a given noise level threshold $\theta >0$, i.e., only if $%
\left\vert \left[
\begin{array}{c}
y \\
u%
\end{array}%
\right] (k)\right\vert >\theta $ for some $k$.   Note that we implicitly assume that there are other inputs such as noise, not shown in Fig \ref{SDfig}, that have some effect on $u$ and $y$ which is what relates to the nonzero noise level $\theta$.  Accordingly, a stealthy
attack will be the case when the attack inputs $d_{a}$ and/or $d_{s}$ can
grow unbounded while maintaining their effect on $u$ and $y$ below the
detection limit. Specifically, if $d$ represents any of $d_{a}$ or $d_{s}$,
then the attack is stealthy if $\limsup_{k\rightarrow\infty}\left\vert d(k)\right\vert
=\infty $ while $\left\vert \left[
\begin{array}{c}
y \\
u%
\end{array}%
\right] (k)\right\vert \leq \theta $ all $k=0,1,2,\dots $.  In the sequel we consider various attack scenarios and analyze the
conditions of their detectability.



\section{Actuator Attacks}

We start with the case when only actuator attacks $d_a$ are present ($d_s=0$%
) and proceed in characterizing their effect on the monitoring vector $%
\left[
\begin{array}{c}
y \\
u%
\end{array}%
\right]$. Towards this end, let $P$ be factored as $P=\tilde{M}^{-1}\tilde{N}%
=NM^{-1}$ where $\tilde{N}, \tilde{M}$ and $N, M$ are left and right coprime
respectively, and consider the controller $K$ with a similar coprime
factorization as $K=\tilde{X}^{-1}\tilde{Y}=YX^{-1}$. The mappings from $d_a$
to $y$ and $u$ are given respectively as $(I-PK)^{-1}P$ and $K(I-PK)^{-1}P$.
Given that $K$ stabilizes $P$, it holds that $\tilde{M}X-\tilde{N}Y=:W$ is a
stable and stably invertible map (unit). Moreover, it can be easily checked
that
\begin{equation}
\left[
\begin{array}{c}
y \\
u%
\end{array}%
\right]= \left[
\begin{array}{c}
X \\
Y%
\end{array}%
\right]W\tilde{N}d_a.  \label{da-yu}
\end{equation}

As $X$ and $Y$ are right coprime and $W$ is a unit, it follows that a
stealthy attack is possible if and only if $\tilde{N}d_a$ is bounded for an
unbounded $d_a$. That is, when $\limsup_{k\rightarrow\infty}\left\vert d_a(k)\right\vert=\infty$
it holds that $\left\Vert \left[
\begin{array}{c}
y \\
u%
\end{array}%
\right]\right\Vert<\infty$ if and only if $\left\Vert \tilde{N}%
d_a\right\Vert<\infty$. The following proposition is a direct consequence of
the previous analysis.

\begin{proposition}
\label{zeros} Let $P$ be a \textquotedblleft tall" system, i.e., the number
of outputs is greater or equal to the number of inputs. Assume further that $%
P(\lambda )$ has no zero on the unit circle $\left\vert \lambda \right\vert
=1$. Then, an (unbounded) actuator stealthy attack is possible if and only if $%
P(\lambda )$ has a non-minimum phase zero other than at $\lambda=0$, i.e., a zero for $0<\left\vert\lambda \right\vert <1$.
\end{proposition}

\begin{proof}
Note that the unstable zeros of $P$ are zeros of $\tilde{N}$. Assuming that $%
P$ is SISO for simplicity with $P(z_0)=0$ where $0<\left\vert z_0\right\vert<1$%
, we have that $\tilde{N}(z_0)=0$ and consequently any input $%
d_a(k)=\epsilon {z_0}^{-k}$ will lead via Equation \ref{da-yu} to $%
\left\Vert \left[
\begin{array}{c}
y \\
u%
\end{array}%
\right]\right\Vert<\epsilon C_0$ where the constant $C_0>0$ depends on the
closed loop maps. For example, $C_0$ could be taken as $C_0=\left\Vert
\left[
\begin{array}{c}
(I-P(\lambda)K(\lambda))^{-1}P(\lambda) \frac{1}{1-(\lambda/z_0)} \\
K(\lambda)(I-P(\lambda)K(\lambda))^{-1}P(\lambda) \frac{1}{1-(\lambda/z_0)}%
\end{array}%
\right]\right\Vert. $ Thus, if $\epsilon$ is small enough, e.g., $0<\epsilon<%
\frac{\theta}{C_0}$, the input remains undetected. In the case where $P$ is
MIMO, the same arguments apply for inputs of the form $d_a(k)=\epsilon d_0{%
z_0}^{-k}$ where $d_0$ is the zero direction of $z_0$ which can be chosen
with $\left\vert d_0\right\vert=1$ .

To prove the reverse, note that if $P$ has no unstable zeros, then the same
holds for $\tilde{N}$ and thus $\left\Vert \tilde{N}d_a\right\Vert<\infty$
implies that $\left\Vert d_a\right\Vert<\infty$, so no stealth attacks are
possible.
\end{proof}

\begin{remark}
\label{rem:01}We remark here that if $P$ has zeros on the boundary $%
\left\vert \lambda \right\vert =1$ with no multiplicity but no other
unstable zeros (other than at $\lambda=0$,) then stealth attacks are not possible. Indeed, if $z_{0}$ is
a simple zero with $\left\vert z_{0}\right\vert =1$, then the corresponding
input that can be masked (``zeroed out") is of the form $d_{a}(k)=\epsilon d_{0}{z_{0}}^{-k}$
which is bounded with $\left\vert d_{a}(k)\right\vert <\epsilon $, and becomes undetected for small enough $\epsilon$.  But this case is uninteresting, as the disturbance has a level of noise (which can be taken care by any reasonably robust controller.)  On the
other hand, if there are multiplicities, stealthy attacks are possible. For
example, if $P$ is SISO and $z_{0}=1$ is a zero with multiplicity $2$, then
an unbounded input of the form $d_{a}(k)=\epsilon k$, $k=0,1,\dots $ remains
undetected for small enough $\epsilon .$ More generally, in the MIMO case
when a zero at the boundary has multiplicity, one has to check the
Smith-McMillan form of $P(\lambda )$ for invariant factors with multiplicity
corresponding to these zeros: stealthy attacks are possible if and only if
there are such factors.
\end{remark}

\begin{remark}
\label{rem:02}
When there is a zero of $P$ at $\lambda=0$ there is no corresponding  (causal) input signal to be ``zeroed out."
\end{remark}

The case when $P$ is \textquotedblleft fat", i.e. when the number of outputs
$y$ is less than the number of inputs $u$, is always conducive to stealthy
attacks as one input can mask the effect of the other. Indeed, consider a
two input one output $P=[P_{1}~P_{2}]$; the effect of attacks at the
individual control channels $d_{a1}$ and $d_{a2}$ on the output $y$ is $%
y=P_{1}d_{a1}+P_{2}d_{a2}+[P_{1}~P_{2}]u$ and thus, picking for example, $%
d_{a2}=-P_{2}^{-1}P_{1}d_{a1}$ with $d_{a1}$ arbitrary and unbounded leads
to $y=[P_{1}~P_{2}]u$, i.e. complete masking of the attacks.  \footnote{Strictly speaking, $P_2^{-1}$ may not exist if $P_2$ is strictly proper , i.e., $P_2$ has a zero at $\lambda=0$; but one can always pick $d_{a1}(\lambda)=\lambda \bar{d}_{a1}(\lambda) $ with $\bar{d}_{a1} $ unbounded and  make $(P_{2}^{-1}P_{1}d_{a1})(\lambda)$ meaningful. }




\section{Sensor Attacks}

The case of sensor only attack $d_s\ne 0, d_a=0$ can be viewed in a
similar spirit. In particular, by considering coprime factorizations for $P$
and $K$ as before, the effect of $d_s$ on the monitor vector is as
\begin{equation}
\left[
\begin{array}{c}
y \\
u%
\end{array}%
\right]=\left[
\begin{array}{c}
(I-PK)^{-1} \\
K(I-PK)^{-1}%
\end{array}%
\right]d_s=\left[
\begin{array}{c}
X \\
Y%
\end{array}%
\right]W\tilde{M}d_s.  \label{ds-yu}
\end{equation}

Therefore, using the same rationale as in the previous case, we can claim
that an attack is detectable if and only if there are no $d_s$ with $%
\left\Vert d_s\right\Vert=\infty$ and $\left\Vert \tilde{M}%
d_s\right\Vert<\infty$. This in turn means that attacks are detectable if
and only if $\tilde{M}$ has no unstable zeros, which is equivalent that $P$
is a stable system. More specifically, we have the following which can be
proved as in the Proposition \ref{zeros}.

\begin{proposition}
\label{poles} Assume that $P(\lambda)$ has no pole on the unit circle $%
\left\vert \lambda\right\vert=1$. Then, a sensor stealthy attack is possible
if and only if $P(\lambda)$ has a pole with $0<\left\vert
\lambda\right\vert<1$, i.e., an unstable pole other than $\lambda=0$.
\end{proposition}

Regarding poles of $P(\lambda)$ on the boundary $\left\vert
\lambda\right\vert=1$ similar remarks hold as in the actuator attack case.
Namely, if these poles are simple then there is no stealthy attack. If they
have multiplicities, then their multiplicities in the corresponding
invariant factors in the Smith-McMillan form determine whether stealthy
attacks are possible.



\section{Coordinated Actuator Sensor Attacks}

In the case when a coordination of actuator and sensor attack is possible,
stealthy attacks are always possible even in the case where $P$ is stable
and minimum phase. Indeed, in this case the effect of $d_a$ can be
completely masked by canceling its effect at the output via $d_s$: just pick
$d_s=-Pd_a$ with $d_a$ arbitrary and unbounded, then $y=Pu$. Therefore,
unless there are outputs that are not attacked, this situation is not of
interest as there is no hope to detect the attack. If there are such attack-free outputs, then the problem reverts to the
actuator only attack case, with these outputs used for analysis and design.  As a consequence, in the sequel we consider the actuator only attack case where the secure sensor outputs are assumed to provide an observable continuous time system $P_c$.

\section{Dual Rate Control}

\begin{figure}[t]
\centering
\includegraphics
[width=3in]
{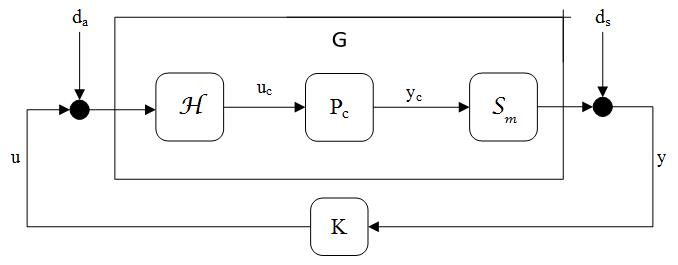}
\caption{A dual rate SD system}
\label{MRfig}
\end{figure}

In this paper, we focus on a particular MR scheme that allows attacks to be
detected by ensuring that there are no relevant unstable zeros in the lifted
system. This scheme is simpler to the single rate one with periodic
controller obtained in the context of gain margin maximization in \cite%
{frageo}. More specifically, we consider the SD scheme of Figure \ref{MRfig}
(temporarily without any disturbances) where the output is sampled with
period $T/m$ where $m$ is a sufficiently large integer, i.e., $y(k)=(%
\mathcal{S}_my_c)(t):=y_c(kT/m)$. A similar scheme has been used in \cite%
{hagiwara} in the context of strong stabilization. Herein, we provide
certain properties of the unstable zeros of the lifted system that guarantee
detectability of actuator attacks.

To this end, let the corresponding discrete-time system mapping $u$ to $y$
be
\begin{equation*}
G=\mathcal{S}_{m}P_{c}\mathcal{H}.
\end{equation*}
For this MR discrete system we have that
\begin{equation*}
\Lambda ^{m}G=G\Lambda
\end{equation*}%
where $\Lambda $ is the 1-step right shift operator on discrete sequences $%
\{x(k)\}$, i.e., $(\Lambda x)(k+1)=x(k)$ with $(\Lambda x)(0)=0$. Using
standard lifting techniques (e.g., \cite{chenfranbook}) one can obtain a
shift invariant (LTI) description $\tilde{G}$ of the discrete dynamics by
grouping the plant input and output signals as $\tilde{u}(k)=u(k)$ and $%
\tilde{y}(k)= [y^{\prime }_c(kT/m)~y^{\prime }_c((k+1)T/m)\dots y^{\prime
}_c((k+m-1)T/m)]^{\prime }$ (similarly for $\tilde{d}_a$ and $\tilde{d}_s$.)
A state space description for $\tilde{G}$ can be obtained as follows:

Define state space matrices
\begin{equation*}
\!\begin{aligned}
A & :=e^{A_{c}T/m}\in \mathbb{R}^{n\times n},& B& :=\int_{0}^{T/m}e^{A_{c}\tau
}B_{c}\mathrm{d}\tau \in \mathbb{R}^{n\times n_{u}}, \\
C & :=C_{c}\in \mathbb{R}^{n_{y}\times n},& D& :=D_{c}\in \mathbb{R}%
^{n_{y}\times n_{u}}.
\end{aligned}
\end{equation*}
Then%
\begin{equation}
\tilde{G}=\left[
\begin{tabular}{l|l}
$\tilde{A}$ & $\tilde{B}$ \\ \hline
$\tilde{C}$ & $\tilde{D}$%
\end{tabular}%
\right] ,  \label{eq:lifted}
\end{equation}%
where%
\begin{eqnarray*}
\tilde{A} &=&A^{m}\in \mathbb{R}^{n\times n},\tilde{B}%
=\sum_{k=0}^{m-1}A^{k}B\in \mathbb{R}^{n\times n_{u}}, \\
\tilde{C} &=&\left[
\begin{array}{c}
C \\
CA \\
\vdots  \\
CA^{m-1}%
\end{array}%
\right] \in \mathbb{R}^{mn_{y}\times n}, \\
\tilde{D} &=&\left[
\begin{array}{c}
D \\
CB+D \\
\vdots  \\
C\sum_{k=0}^{m-2}A^{k}B+D%
\end{array}%
\right] \in \mathbb{R}^{mn_{y}\times n_{u}}.
\end{eqnarray*}%
Also, it becomes useful to define a discrete-time system $P_{m}:=\left[
\begin{tabular}{l|l}
$A$ & $B$ \\ \hline
$C$ & $D$%
\end{tabular}%
\right] $. This system corresponds to the single-rate sampling and hold
scheme of the original plant $P_{c}$ with a period of $T/m$, i.e., $P_{m}=%
\mathcal{S}_{m}P_{c}\mathcal{H}_{m}$ where $\mathcal{H}_{m}$ is accordingly
generating a continuous signal $u_{c}$ from the discrete $u$ as $u_{c}(t)=(%
\mathcal{H}_{m}u)(t)=u(k)$ for $kT/m\leq t<(k+1)T/m$. It is clear that $P_{m}
$ has the same dimension as $P_{c}$, i.e. it maps $n_{u}$ inputs to $n_{y}$
outputs. Moreover, given that $P_{c}$ holds a controllable and observable
realization, and the sampling is not pathological, it follows that the
inherited realization of $P_{m}$ is also controllable and observable. Based
on our assumptions on the sampling, it is also easily verified that the
realization of $\tilde{G}$ as above is controllable and observable. Let $%
\tilde{M}_{\tilde{G}}$ and $\tilde{N}_{\tilde{G}}$ be the left coprime
factors of $\tilde{G}$. We will use the state-space realization of $\tilde{N}%
_{\tilde{G}}$ as
\begin{equation}
\tilde{N}_{\tilde{G}}=\left[
\begin{tabular}{c|c}
$\tilde{A}+H\tilde{C}$ & $\tilde{B}+H\tilde{D}$ \\ \hline
$\tilde{C}$ & $\tilde{D}$%
\end{tabular}%
\right] ,  \label{eq:n_tilde}
\end{equation} where $H$ is chosen such that $\tilde{A}+H\tilde{C}$ is Schur stable. It is
easy to show that $\tilde{G}$ and $\tilde{N}_{\tilde{G}}$ have the same
non-minimum phase zeros.
\begin{figure}[t]
\centering
\includegraphics
[width=2in]
{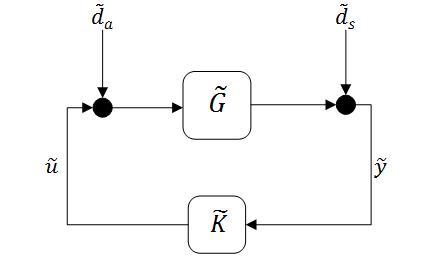}
\caption{The lifted system}
\label{lifted}
\end{figure}
We consider now the closed loop in the lifted domain in Figure \ref{lifted}
where the controller is $\tilde{K}$ and proceed to argue that the lifted
loop is not susceptible to stealthy actuator attacks $\tilde{d}_{a}$, and
thus the original MR loop of Figure \ref{MRfig} is not susceptible either.
To this end, the integer $m$ is chosen such that the following assumptions
are satisfied.

\begin{assumption}
\label{assume:01}The matrix $B$ is full column rank.
\end{assumption}

\begin{assumption}
\label{assume:02}The matrix $\mathcal{O}:\mathcal{=}\left[
\begin{array}{c}
C \\
CA \\
\vdots \\
CA^{m-2}%
\end{array}%
\right] $ is full column rank.
\end{assumption}

The first assumption is standard and holds generically if $B_c$ is full
column rank in the continuous system. The second assumption holds for large
enough $m$, in particular $m=n+1$, if the pair $\left( A,C\right) $ is
observable, which is true as $P_m$ is minimal. It can also hold however,
even with a small $m$ generically.  Also, if
Assumption \ref{assume:02} holds, $\tilde{G}$ is a tall system. Then the
following lemma characterizes the zeros of $\tilde{G}$.

\begin{lemma}
\label{lem:01}Consider the lifted system $\tilde{G}$ as in (\ref{eq:lifted})
together with Assumptions (\ref{assume:01}) and (\ref{assume:02}). Then $%
\tilde{G}$ has at most one non-minimum zero and is located at $\lambda =1$.
\end{lemma}

\begin{proof}
Since $\tilde{N}_{\tilde{G}}$ and $\tilde{G}$ have the same non-minimum
phase zeros, we will prove this lemma for $\tilde{N}_{\tilde{G}}$. Notice
that since $\tilde{N}_{\tilde{G}}$ is tall, $\left\vert \lambda
_{0}\right\vert \leq 1$ is a zero if and only if there exists a vector $\nu
\in \mathbb{R}^{n_{u}}$ such that%
\begin{eqnarray*}
&&\tilde{N}_{\tilde{G}}\left( \lambda _{0}\right) \nu = \\
&&\left[ \lambda _{0}\tilde{C}\left[ I-\lambda _{0}\left( \tilde{A}+H\tilde{C%
}\right) \right] ^{-1}\left( \tilde{B}+H\tilde{D}\right) +\tilde{D}\right]
\nu =0.
\end{eqnarray*}%
Notice that $\left[ I-\lambda _{0}\left( \tilde{A}+H\tilde{C}\right) \right]
^{-1}$ is well-defined as all the eigenvalues of $\tilde{A}+H\tilde{C}$ are
inside the unit circle. Now, let $\xi =\left[ I-\lambda _{0}\left( \tilde{A}%
+H\tilde{C}\right) \right] ^{-1}\left( \tilde{B}+H\tilde{D}\right) \nu $.
Then, pre-multiplying by $\left[ I-\lambda _{0}\left( \tilde{A}+H\tilde{C}%
\right) \right] $ and using $\lambda _{0}\tilde{C}\xi +\tilde{D}\nu =0,$ we
get
\begin{eqnarray}
\lambda _{0}\tilde{C}\xi +\tilde{D}\nu  &=&0,  \label{eq:04} \\
\left( I-\lambda _{0}\tilde{A}\right) \xi -\tilde{B}\nu  &=&0.  \label{eq:05}
\end{eqnarray}%
Pre-multiplying (\ref{eq:04}) by $X$, where $X$ is a matrix $X\in \mathbb{R}%
^{(m-1)n_{y}\times mn_{y}}$ given as
\begin{equation}
X=\left[
\begin{array}{ccccc}
I & -I & 0 & \cdots  & 0 \\
0 & I & -I &  &  \\
\vdots  &  & \ddots  & \ddots  &  \\
0 & \cdots  & 0 & I & -I%
\end{array}%
\right] .  \label{eq:X}
\end{equation}%
we get%
\begin{equation*}
\lambda _{0}X\tilde{C}\xi +X\tilde{D}\nu =\mathcal{O}\left[ \lambda
_{0}\left( I-A\right) \xi -B\nu \right] =0.
\end{equation*}%
Since $\mathcal{O}$ is full column rank by Assumption \ref{assume:02}, it
holds true that%
\begin{equation*}
\lambda _{0}\left( A-I\right) \xi +B\nu =0,
\end{equation*}%
which together with (\ref{eq:05}) gives%
\begin{equation*}
\left[ \left( I-\lambda _{0}\tilde{A}\right) B+\lambda _{0}\left( A-I\right)
\tilde{B}\right] \nu =0.
\end{equation*}%
Simplifying further yields%
\begin{equation*}
\left( 1-\lambda _{0}\right) B\nu =0.
\end{equation*}%
Therefore, if $\nu $ is nonzero then $\lambda _{0}=1$ since, by Assumption %
\ref{assume:01}, $B$ is full column rank.
\end{proof}

According to Lemma \ref{lem:01}, the lifted system, $\tilde{G}$, has no
zeros inside the unit circle. However, it may have a zero at $\lambda =1$.
Based on Proposition \ref{zeros} and Remark \ref{rem:01}, an (unbounded)
actuator stealthy attack will not possible if $\lambda =1$ is zero of $%
\tilde{G}$ with multiplicity of at most one. Indeed, this is the case as it
is proved in the following theorem:

\begin{theorem}
Consider the dual rate SD scheme as in Figure \ref{lifted}. Then, there does
not exist any (unbounded) actuator stealthy attack if Assumptions \ref%
{assume:01} and \ref{assume:02} are met.
\end{theorem}

\begin{proof}
As discussed before, we need to show that $\lambda =1$ is a zero of $\tilde{G%
}$ or equivalently $\tilde{N}_{\tilde{G}}$ with the multiplicity of at most
one. It can be argued that (\cite{dahlehb}-Section 6.5) $\lambda =1$ is a
zero of algebraic multiplicity greater than one if and only if the matrix $%
T:=\left[
\begin{array}{cc}
\tilde{N}_{\tilde{G}}\left( 1\right)  &  0 \\
\frac{\mathrm{d}}{\mathrm{d}\lambda }\tilde{N}_{\tilde{G}}\left( \lambda
\right) |_{\lambda =1} & \tilde{N}_{\tilde{G}}\left( 1\right)
\end{array}%
\right] $ has a right null chain; that is, there exists a vector $\nu =\left[
\begin{array}{c}
\nu _{1} \\
\nu _{2}%
\end{array}%
\right] $, with $\nu _{1}\neq 0$, such that $T\nu =0$. By the way of
contradiction, we will show that if $T\nu =0$ then $\nu _{1}=0$. Direct
calculations show that if $T\nu =0$ then%
\begin{equation}
\left[ \tilde{C}\left[ I-\left( \tilde{A}+H\tilde{C}\right) \right]
^{-1}\left( \tilde{B}+H\tilde{D}\right) +\tilde{D}\right] \nu _{1}=0,
\label{eq:06}
\end{equation}%
\begin{align}
& \left[ \tilde{C}\left[ I-\left( \tilde{A}+H\tilde{C}\right) \right]
^{-2}\left( \tilde{B}+H\tilde{D}\right) \right] \nu _{1}  \notag \\
& +\left[ \tilde{C}\left[ I-\left( \tilde{A}+H\tilde{C}\right) \right]
^{-1}\left( \tilde{B}+H\tilde{D}\right) +\tilde{D}\right] \nu _{2}=0.
\label{eq:07}
\end{align}%
Define,%
\begin{eqnarray*}
\xi _{1} &=&\left[ I-\left( \tilde{A}+H\tilde{C}\right) \right] ^{-1}\left(
\tilde{B}+H\tilde{D}\right) \nu _{1}, \\
\xi _{2} &=&\left[ I-\left( \tilde{A}+H\tilde{C}\right) \right] ^{-1}\left[
\xi _{1}+\left( \tilde{B}+H\tilde{D}\right) \nu _{2}\right] .
\end{eqnarray*}%
Pre-multiplying $\xi _{1}$ and $\xi _{2}$ by $\left[ I-\left( \tilde{A}+H%
\tilde{C}\right) \right] $ and grouping terms we get%
\begin{eqnarray}
\left( I-\tilde{A}\right) \xi _{1}-\tilde{B}\nu _{1} &=&H\left( \tilde{C}\xi
_{1}+\tilde{D}\nu _{1}\right) ,  \label{eq:08} \\
-\xi _{1}+\left( I-\tilde{A}\right) \xi _{2}-\tilde{B}\nu _{2} &=&H\left(
\tilde{C}\xi _{2}+\tilde{D}\nu _{2}\right) .  \label{eq:09}
\end{eqnarray}%
From (\ref{eq:06})-(\ref{eq:09}),%
\begin{eqnarray}
\tilde{C}\xi _{1}+\tilde{D}\nu _{1} &=&0,  \label{eq:10} \\
\tilde{C}\xi _{2}+\tilde{D}\nu _{2} &=&0,  \label{eq:11} \\
\left( I-\tilde{A}\right) \xi _{1}-\tilde{B}\nu _{1} &=&0,  \label{eq:12} \\
-\xi _{1}+\left( I-\tilde{A}\right) \xi _{2}-\tilde{B}\nu _{2} &=&0.
\label{eq:13}
\end{eqnarray}%
Furthermore, pre-multiplying (\ref{eq:10}) and (\ref{eq:11}) gives%
\begin{eqnarray*}
X\tilde{C}\xi _{1}+X\tilde{D}\nu _{1} &=&\mathcal{O}\left[ \left( I-A\right)
\xi _{1}-B\nu _{1}\right] =0, \\
X\tilde{C}\xi _{2}+X\tilde{D}\nu _{2} &=&\mathcal{O}\left[ \left( I-A\right)
\xi _{2}-B\nu _{2}\right] =0,
\end{eqnarray*}%
where $X$ is as in (\ref{eq:X}), which in turn imply%
\begin{eqnarray}
\left( I-A\right) \xi _{1}-B\nu _{1} &=&0,  \label{eq:14} \\
\left( I-A\right) \xi _{2}-B\nu _{2} &=&0.  \label{eq:15}
\end{eqnarray}%
Eliminating $\xi _{2}$ between (\ref{eq:13}) and (\ref{eq:15}), we get%
\begin{equation*}
-\left( I-A\right) \xi _{1}-\left[ \left( I-A\right) \tilde{B}-\left( I-%
\tilde{A}\right) B\right] \nu _{2}=0.
\end{equation*}%
Notice that $\left( I-A\right) \tilde{B}-\left( I-\tilde{A}\right) B=0$ and
hence the last equation implies%
\begin{equation*}
\left( I-A\right) \xi _{1}=0
\end{equation*}%
which in turn, together with (\ref{eq:14}), implies $B\nu _{1}=0$. By
Assumption \ref{assume:01}, $B\nu _{1}=0$ implies $\nu _{1}=0$ and this
completes the proof.
\end{proof}

As a final comment from the previous analysis, we offer conditions when $%
\tilde{G}$ has a zero $\lambda=1$. We note that, as proved in the previous
theorem, these zeros are not a problem since they cannot generate stealthy
attacks.

\begin{proposition}
\label{prop:01}Let $P_{c}$ be \textquotedblleft tall." Then $\tilde{G}$ has
a zero at $\lambda =1$ if and only if $P_{m}$ does.
\end{proposition}

\begin{proof}
Suppose $\tilde{G}$ has a zero at $\lambda =1$. Then, there exist vectors $%
\xi $ and $\nu $, at least one of them nonzero, such that (\ref{eq:04}) and (%
\ref{eq:05}) hold for $\lambda _{0}=1$. In particular, from (\ref{eq:04}) we
get%
\begin{equation}
C\xi +D\nu =0.  \label{eq:20}
\end{equation}%
Furthermore, pre-multiplying (\ref{eq:04}) by $X$ results in $\mathcal{O}%
\left[ \left( I-A\right) \xi -B\nu \right] =0$ which in turn implies
\begin{equation}
\left( I-A\right) \xi -B\nu =0.  \label{eq:21}
\end{equation}%
(\ref{eq:20}) and (\ref{eq:21}) imply that $P_{m}$ has a zero at $\lambda =1$%
.

Conversely, if $P_{m}$ has a zero at $\lambda =1$,
\begin{equation}
\left[
\begin{array}{cc}
I-A & -B \\
C & D%
\end{array}%
\right] \left[
\begin{array}{c}
\xi \\
\nu%
\end{array}%
\right] =0,  \label{eq:18}
\end{equation}
for some $\xi $ and $\nu $. Pre-multiplying it by%
\begin{equation}
\left[
\begin{array}{cc}
\sum_{k=1}^{m-1}A^{k} & 0 \\
0 & I \\
-C & I \\
-C-CA & I \\
\vdots &  \\
-C\sum_{k=0}^{m-2}A^{k} & I%
\end{array}%
\right]  \label{eq:19}
\end{equation}%
gives $\left[
\begin{array}{cc}
I-\tilde{A} & -\tilde{B} \\
\tilde{C} & \tilde{D}%
\end{array}%
\right] \left[
\begin{array}{c}
\xi \\
\nu%
\end{array}%
\right] =0$. That is, $\tilde{G}$ has a zero at $\lambda =1$.
\end{proof}

\begin{proposition}
Let $P_{c}$ be \textquotedblleft fat." Then $\tilde{G}$ has always a zero at
$\lambda =1$.
\end{proposition}

\begin{proof}
The proof relies on the fact that since $P_{c}$ or equivalently $P_{m}$ is
fat, there always exist two vectors $\xi $ and $\nu $ with at least one of
them nonzero such that (\ref{eq:18}) holds. Then, the rest of the proof
follows similarly to that of the converse part of Proposition \ref{prop:01}.
\end{proof}

\begin{remark}
We would like to point out that an equivalent way of obtaining the same
results, i.e., ability to detect zero attacks, is to hold the control input
longer rather than sampling the output faster. That is, if we consider a
dual rate system where the hold operates with a period of $mT$ while the
output is sampled with $T$, then the corresponding lifted system will enjoy
the same properties as before in terms of unstable zeros. Obviously, the
(nominal) controller performance will be reduced as the control is slower.
On the other hand, there is a potential benefit of lower cost of actuation
in this case.
\end{remark}

\section{Conclusion}

We presented a simple dual rate sampled data scheme which guarantees detectability of actuator and/or sensor attacks, if a secure output that maintains observability of the open loop modes  is available.  The main observation is that the sampled data nature in the implementation of the cybephysical system cannot be ignored as sampling can generate additional vulnerabilities due to the extra unstable zeros it may introduce, particularly if high rates are necessary to achieve certain performance level.  The proposed method takes care of this issue by the use of multirate sampling that ensures that  zeros exist only in harmless locations in the lifted domain.  We gave certain precise conditions on the detectability of stealthy attacks in terms of the open loop unstable poles and zeros and showed how the vulnerabilities can be eradicated by the use of the dual rate scheme. 

Several other possibilities can be studied in this context. The use of asynchronous sampling (e.g., \cite{vou94,toiv98}) can provide alternative ways to detect stealthy attacks; or even the network's random delays can be helpful in that respect; the speed of detecting however needs to be brought into consideration, even if the attack is detectable.  The methods of generalized holds \cite{kab} are also relevant as they move zeros, and with careful analysis of their robustness properties (e.g., \cite{freu,dulbook}) can provide acceptable and simple solutions as well.  All of these are subjects of current investigations by the authors and are documented in forthcoming publications.


\begin{thebibliography}{}


\bibitem{Ning}
Y. Liu, M. K. Reiter, and P. Ning. "False data injection attacks against state estimation in electric power grids". \emph{ACM Conference
on Computer and Communications Security}, Chicago, IL, USA, Nov. 2009, pp. 21,32.

\bibitem{sandberg1}
H. Sandberg, A. Teixeira, and K. H. Johansson. "On security indices for state estimators in power networks". \emph{First Workshop on Secure Control Systems,} Stockholm, Sweden, 2010.

\bibitem{sandberg2}
A. Teixeira, S. Amin, H. Sandberg, K. H. Johansson, and S. Sastry. "Cyber security analysis of state estimators in electric power systems". \emph{IEEE Conference on Decisions and Control,} Atlanta, GA, USA, Dec. 2010.

\bibitem{tabuada}
H. Fawzi, P. Tabuada, P. Diggavi. "Secure Estimation and Control for Cyber-Physical Systems Under Adversarial Attacks". \emph{IEEE Transactions on Automatic Control,} vol.59, no.6, pp.1454,1467, June 2014.

\bibitem{hadj}
S. Sundaram, C.N. Hadjicostis. "Distributed Function Calculation via Linear Iterative Strategies in the Presence of Malicious Agents". \emph{IEEE Transactions on Automatic Control}, vol.56, no.7, pp.1495,1508, July 2011.

\bibitem{sinopoli}
Y. Mo and B. Sinopoli. "Secure control against replay attacks". \emph{Allerton Conf. on Communications, Control and Computing,} Monticello, IL, USA, Sept. 2010, pp. 911-918.

\bibitem{newsino}
R. Chabukswar, Y. Mo, and B. Sinopoli. "Detecting Integrity Attacks on SCADA Systems". \emph{IEEE Transactions on Control Systems Technology: a Publication of the IEEE Control Systems Society}, vol.22, no.4, pp.1396,1407, July 2014.

\bibitem{bullo}
F. Pasqualetti and F. Dorfler and F. Bullo. "Attack Detection and Identification in Cyber-Physical Systems". \emph{IEEE Transactions on Automatic Control,} vol.58, no.11, pp.2715-2729, 2013.

\bibitem{sandberg3}
A. Teixeira, I. Shames, H. Sandberg, and K. H. Johansson. "Revealing Stealthy Attacks in Control Systems". \emph{Allerton Conf. on Communication, Control, and Computing,} Allerton, IL, USA, 2012.


\bibitem{chenfranbook}
T. Chen and B. Francis. \emph{Optimal Sampled-Data Control Systems}. (errerta). Springer, 1995.


\bibitem{hagiwara}
T. Hagiwara, M. Araki, "Design of a stable state feedback controller based on the multirate sampling of the plant output". \emph{IEEE Transactions on Automatic Control,} vol.33, no.9, pp.812,819, Sep 1988.

\bibitem{voudah94}
P.G. Voulgaris, M.A. Dahleh and L.S. Valavani. "$H_{\infty}$ and $H_2$ optimal controllers for periodic and multirate systems". \emph{Automatica,} vol. 30, no. 2, pp. 252-263, 1994.

\bibitem{voubam93}
P.G. Voulgaris, B. Bamieh. "Optimal $H_{\infty}$ and $H_2$ control of hybrid multirate systems". \emph{Systems and Control Letters,} no. 20, pp. 249-261, 1993.


\bibitem{chenqi}
T. Chen., L. Qiu. "$H_{\infty}$ design of general multirate sampled-data control systems". \emph{IEEE Transactions on Automatic Control,} vol.39, no.12, pp.2506-2511, 1994.


\bibitem{qichen}
L. Qiu, T. Chen. "$H_2$ optimal design of multirate sampled-data systems". \emph{Automatica,} Volume 30, Issue 7, July 1994, Pages 1139-1152.

\bibitem{cola}
P. Colaneri, R. Scattolini, N. Schiavoni. "Stabilization of multirate sampled-data linear systems". \emph{Automatica,} Volume 26, Issue 2, March 1990, Pages 377-380.

\bibitem{meyer}
D.G. Meyer.  "A parametrization of stabilizing controllers for multirate sampled-data systems". \emph{IEEE Transactions on Automatic Control,} vol.35, no.2, pp.233,236, Feb 1990.

\bibitem{khargo}
R. Ravi, P.P. Khargonekar, K.D. Minto,  C.N. Nett. "Controller parametrization for time-varying multirate plants". \emph{IEEE Transactions on Automatic Control,} vol.35, no.11, pp.1259,1262, Nov 1990.

\bibitem{frageo}
B.A. Francis,, T.T. Georgiou.  "Stability theory for linear time-invariant plants with periodic digital controllers". \emph{IEEE Transactions on Automatic Control,} vol.33, no.9, pp.820,832, Sep 1988.

\bibitem{dahlehb}
M. A.  Dahleh and I. Diaz-Bobillo. \emph{Control of Uncertain Systems: A Linear Programming Approach}. Prentice-Hall, 1995.

\bibitem{saadat}
H. Saadat. \emph{Power System Analysis}. McGraw-Hill Companies, 2002.


\bibitem{Johansson}
K. Johansson. "The quadruple-tank process: a multivariable laboratory process with an adjustable zero". \emph{IEEE Transactions on Control
Systems Technology,} vol. 8, no. 3, pp. 456"1¤7465, May 2000.

\bibitem{vou94}
P.G. Voulgaris. "Control of asynchronous sampled-data systems". \emph{IEEE Transactions on Automatic Control,} vol. 39, no. 7, pp. 1451-1455, July 1994.


\bibitem{toiv98}
M. F. Sagfors, H. T. Toivonen. "$H_{\infty}$ and LQG control of asynchronous sampled-data systems". \emph{Automatica,} Volume 33, Issue 9, September 1997, Pages 1663-1668.

\bibitem{kab}
P.T. Kabamba. "Control of linear systems using generalized sampled-data hold functions". \emph{IEEE Transactions on Automatic Control,} vol.32, no.9, pp.772,783, Sep 1987.

\bibitem{freu}
J.S Freudenberg, R.H. Middleton, J.H. Braslavsky. "Robustness of zero shifting via generalized sampled-data hold functions". \emph{IEEE Transactions on Automatic Control,} vol.42, no.12, pp.1681,1692, Dec 1997.

\bibitem{dulbook}
G.E. Dullerud. \emph{Control of Uncertain Sampled-Data Systems}. Birkhauser 1996.

\end{thebibliography}
\end{document}